\documentclass[final,5p,times,twocolumn]{elsarticle}

\usepackage{amssymb}
\usepackage{paralist}
\usepackage{amsmath}
\usepackage{amsthm}
\usepackage{subfigure}
\usepackage{color}
\usepackage[lined,ruled,boxed,commentsnumbered]{algorithm2e}

\newtheorem{lem}{Lemma}
\newtheorem{thm}{Theorem}

\journal{Computer Communications}

\begin{document}

\begin{frontmatter}

\title{SACRM: Social Aware Crowdsourcing with Reputation Management in Mobile Sensing}

\author[label1,label2]{Ju Ren\corref{cor1}}
\author[label1]{Yaoxue Zhang}
\author[label2]{Kuan Zhang}
\author[label2]{Xuemin (Sherman) Shen}

\cortext[cor1]{Ju Ren is the corresponding author. He is a PhD candidate in the Central South University, Changsha, China. And he is now a visiting scholar in the University of Waterloo, Ontario, Canada. His phone number is (+1)~519-781-2715 and email address is renjullcsu@gmail.com.}
\address[label1]{School of Information Science and Engineering, Central South University, Changsha, China 410083}
\address[label2]{Dept. of Electrical and Computer Engineering, University of Waterloo, Waterloo, Ontario, Canada N2L 3G1}

\begin{abstract}
Mobile sensing has become a promising paradigm for mobile users to obtain information by task crowdsourcing. However, due to the social preferences of mobile users, the quality of sensing reports may be impacted by the underlying social attributes and selfishness of individuals. Therefore, it is crucial to consider the social impacts and trustworthiness of mobile users when selecting task participants in mobile sensing. In this paper, we propose a \textbf{\underline{S}}ocial  \textbf{\underline{A}}ware \textbf{\underline{C}}rowdsourcing with \textbf{\underline{R}}eputation \textbf{\underline{M}}anagement (SACRM) scheme to select the well-suited participants and allocate the task rewards in mobile sensing. Specifically, we consider the social attributes, task delay and reputation in crowdsourcing and propose a participant selection scheme to choose the well-suited participants for the sensing task under a fixed task budget. A report assessment and rewarding scheme is also introduced to measure the quality of the sensing reports and allocate the task rewards based the assessed report quality. In addition, we develop a reputation management scheme to evaluate the trustworthiness and cost performance ratio of mobile users for participant selection. Theoretical analysis and extensive simulations demonstrate that SACRM can efficiently improve the crowdsourcing utility and effectively stimulate the participants to improve the quality of their sensing reports.
\end{abstract}


\end{frontmatter}

\section{Introduction}

We have witnessed recently the dramatic proliferation of mobile computing devices such as smartphones and tablet computers~\cite{20}. Since these devices are generally equipped with a set of versatile sensors, mobile sensing (also known as participatory sensing or urban sensing) has emerged as a new horizon for ubiquitous sensing~\cite{23}. In a typical mobile sensing application~\cite{1,2}, a data requester first publishes a sensing task to crowdsource, and then selects a number of mobile users interested in the task to collect the desired data. Once the participants finish the sensing task, they submit their sensing reports to the data requester and earn their task rewards. Such a new paradigm of information collection brings great benefits (e.g., efficiency and low cost) and also challenging issues by task crowdsourcing~\cite{24}.

One of the challenging issues in mobile sensing is to select participants for crowdsourced sensing tasks. A few of research efforts have been invested to address the participant selection problem. Reddy et. al.~\cite{4} develop a recruitment framework to select well-suited participants for sensing tasks based on the spatio-temporal availability and personal reputation. They highlight participant selection should highly depend on the location and time availability and trustworthiness of the participants. However, little attention has been paid to the underlying social attributes of mobile users (e.g. interests, living area), which are critical for task crowdsourcing~\cite{25,30,28}, especially for participant selection. For a specific sensing task, it generally has a set of interested social attributes, and a large social attribute overlap between the task and mobile user indicates a potential matching and high task quality. For instance, if the published task is ``\textit{Find the cheapest Coca Cola in the Waterloo city}", the mobile users whose social attributes include ``\textit{Shopping}'' and ``\textit{Waterloo}'' might be preferred to be recruited in the task. While the published task changes to ``\textit{Find an unoccupied basketball court in the University of Toronto}", the mobile users whose social attributes include ``\textit{Sporting}" and ``\textit{Toronto}" should be preferred. Therefore, it is of great significance to consider the impact of social attributes on crowdsourcing, especially on the participant selection.

Another challenge in mobile sensing is to evaluate the trustworthiness of sensing reports and participants, and fairly allocate task rewards. In the presence of malicious participants, mobile crowdsensing is vulnerable to various types of attacks, e.g., denial-of-service attack~\cite{35} and data pollution attack~\cite{15,1}, etc. Reputation system is a promising technique and has been widely used in trustworthiness evaluation for mobile sensing~\cite{16,29}. Wang et. al.~\cite{1} propose a reputation framework to evaluate the trustworthiness of sensing reports and participants. Huang et. al.~\cite{15} employ the Gompertz function to compute the device reputation score and evaluate the trustworthiness of contributed data. However, most of the exiting works only focus on trustworthiness evaluation for participants, without adjusting their task rewards based on the quality of their sensing reports. Such that, the malicious users can still earn enough task rewards before their reputation goes to a low value. Therefore, in order to defend this attack and economically stimulate participants' contributions, it is crucial to adaptively allocate the task rewards to the participants according to their sensing report quality.

In this paper,  we propose a  \textbf{\underline{S}}ocial  \textbf{\underline{A}}ware \textbf{\underline{C}}rowdsourcing with \textbf{\underline{R}}eputation \textbf{\underline{M}}anagement (SACRM) scheme to select the well-suited participants and allocate  task rewards in mobile sensing. Compared with the existing works, we synthetically consider the social attributes, task delay and personal reputation in mobile sensing and define a utility function to quantify the effect of the three factors on crowdsourcing. The major contributions of our work are four folds.
\begin{asparaitem}
	\item We propose a participant selection scheme to choose the well-suited participants for sensing tasks and maximize the crowdsourcing utility under a fixed task budget. The proposed scheme consists of two participant selection algorithms, a dynamic programming algorithm to achieve the optimal solution and a fully polynomial time approximation algorithm to achieve the (1-$\epsilon$)-approximate solution.
	\item We propose a report assessment and rewarding scheme to measure the quality of sensing reports and allocate task rewards. Both of the report veracity and report delay are considered as two quality metrics for report assessment. And the task rewards are allocated according to the assessment results.
	\item We develop a reputation management scheme to evaluate the trustworthiness and cost performance ratio of mobile users for participant selection, which can stimulate participants to improve their report quality.
	\item We theoretically analyze the performance of the proposed participant selection algorithms. Extensive simulations demonstrate the effectiveness and efficiency of the SACRM scheme.
\end{asparaitem}

The remainder of the paper is organized as follows. Related works are reviewed in Section~\ref{sec2}. In Section~\ref{sec3}, we provide an overview of the system model and design goals. Section~\ref{sec4} presents the details of the proposed SACRM scheme. The theoretical analysis of SACRM is described in Section~\ref{sec5}. We evaluate the performance of SACRM by extensive simulations in Section~\ref{sec6}. Finally, Section~\ref{sec7} concludes the paper and introduces our future work.

\section{Related Work}
\label{sec2}
As an emerging information collection mechanism, crowdsourcing has been extensively studied in mobile sensing. Most of the related works focus on studying the incentive mechanisms to stimulate the participation of mobile users for crowdsourcing~\cite{2,5,6,9,17,10}.

Dynamic pricing is an effective incentive mechanism widely used in mobile sensing~\cite{2,5,6,7}. Yang et. al.~\cite{2} propose two incentive mechanisms to stimulate mobile users' participation respectively for platform-centric and user-centric mobile sensing. For the platform-centric model, they present a Stackelberg game~\cite{add_game} based incentive mechanism to maximize the utility of the platform. For the user-centric model, they design an auction-based incentive mechanism that is proved to be computationally efficient, individually rational, profitable and truthful. Jaimes et. al.~\cite{5} propose a recurrent reverse auction incentive mechanism using a greedy algorithm to select a representative subset of users according to their locations under a fixed budget. In~\cite{6}, the authors develop and evaluate a reverse auction based dynamic pricing incentive mechanism to stimulate mobile users' participation and reduce the incentive cost. Besides the dynamic pricing mechanism, personal demand and social relationship are introduced into the incentive mechanism study~\cite{9,19,10}. Luo et. al.~\cite{9} link the incentive to personal demand for consuming compelling services. Based on the demand principle, two incentive schemes, called Incentive with Demand Fairness (IDF) and Iterative Tank Filling (ITF), are proposed to maximize fairness and social welfare, respectively.

The majority of the existing incentive mechanisms are beneficial to stimulate the user participation, however, data assessment and reputation management are desired and critical to evaluate the trustworthiness of sensing data and mobile users~\cite{18,11,12,13,14,15}. Zhang et. al.~\cite{18} propose a robust trajectory estimation strategy, called TrMCD, to alleviate the negative influence of abnormal crowdsourced user trajectories and identify the normal and abnormal users, as well as to mitigate the impact of the spatial unbalanced crowdsourced trajectories. Huang et. al.~\cite{15} employ the Gompertz function~\cite{add_function} to compute the device reputation score and evaluate the trustworthiness of the contributed data. Since the reputation scores associated with the specific contributions can be used to identify the participants, privacy issues are highlighted in the reputation system design of mobile sensing~\cite{1,11,14}. Wang et. al.~\cite{1} propose a privacy-preserving reputation framework to evaluate the trustiness of the sensing reports and the participants based on the blind signatures. Christin et. al.~\cite{11} propose an anonymous reputation framework, called as IncogniSense, which generates periodic pseudonyms by blind signature and transfers reputation between these pseudonyms.

Recently, participant selection has been studied to achieve the optimal crowdsourcing utility~\cite{4,10}. Reddy et. al.~\cite{4} develop a recruitment framework to enable the data requester to identify well-suited participants for the sensing task based on geographic and temporal availability as well as the participant reputation. The proposed recruitment system approximately maximizes the coverage over a specific area and time period under a limited campaign budget with a greedy algorithm. Amintoosi et. al.~\cite{10} propose a recruitment framework for social participatory sensing to identify and select suitable and trustworthy participants in the friend circle, by leveraging the multihop friendship relations. However, they do not consider the social attributes of mobile users and adaptive rewards allocation, which play a significant role in crowdsourcing design.

\section{System Model and Design Goals}
\label{sec3}

\subsection{System Model}
\label{sec.systemmodel}
We consider a typical mobile sensing system, which is applied in~\cite{1,2} and illustrated as Fig.~\ref{fig.networkmodel}. The system consists of a mobile sensing application \textit{platform} and a large number of \textit{mobile users}. The application platform generally resides in the cloud and consists of multiple sensing servers, and the mobile users connect to the platform through WiFi or cellular network. Each mobile user can publish his\footnote{No specific gender here, and the same applies in the following paper.} sensing task on the platform, called as \textit{data requester}. And the users who are finally assigned the sensing task are called as \textit{participants}. We describe a complete mobile sensing process as follows.

At first, a data requester has a sensing task and task requirements (e.g., task deadline and task budget), and publishes it on the platform to recruit mobile users to finish it (Step 1). The mobile users, who are interested in participating in the sensing tasks, then estimate the cost and expected delay to finish this task and apply to participate in the task with their application information (e.g., bid price, expected task delay). (Step 2). Then, the data requester chooses a subset of applicants to take this sensing task based on their application information (Step 3). The participants collect the required data information and report them to the platform. The reported data are processed by the sensing servers and then provided to the data requester (Step 4). After that, the data requester would assess the quality of sensing reports and gives a feedback (e.g., sensing reports evaluation, allocated rewards, reputation evaluation), to the platform (Step 5). Finally, the platform processes the feedback (e.g., reputation update), and then returns it to the participants (Step 6).

All the procedures in the mobile sensing system are involved in the crowdsourcing and have attracted a number of research efforts. In this paper, we particularly focus on three key issues, i.e., how to choose the applicants to take the sensing task in Step 3, and how to evaluate the sensing reports and reputation, and develop an adaptive rewarding scheme in Step 5, as well as the feedback processing (e.g., reputation update and management) in Step 6.
\begin{figure}[!t]
\includegraphics[width=0.5\textwidth]{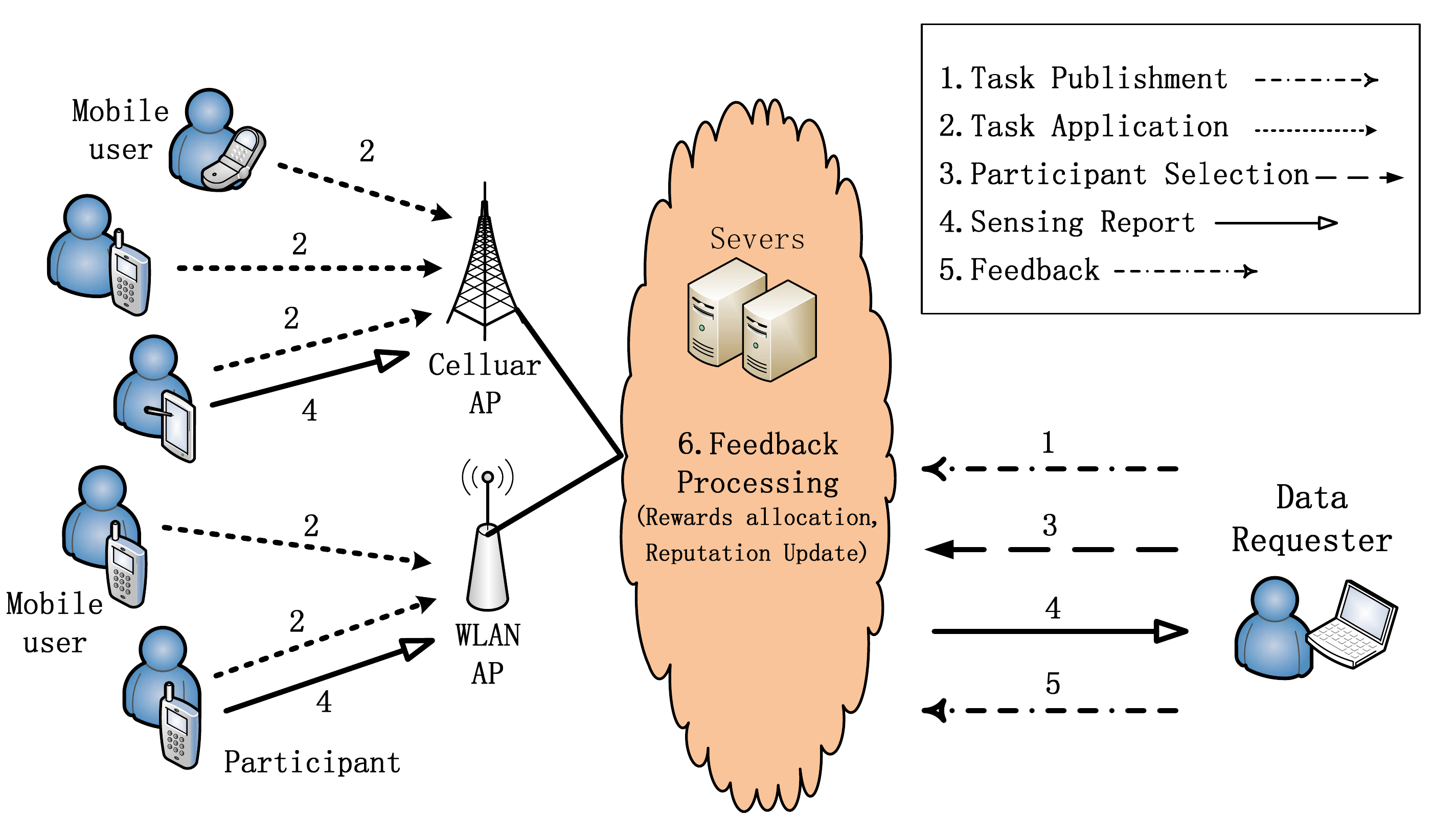}
\caption{Architecture of Mobile Sensing System}
\label{fig.networkmodel}
\end{figure}

\subsection{Design Goals}
The proposed SACRM aims to select the well-suited participants for a specific sensing task and adaptively reward the participants based on the quality of their sensing reports. More specifically, the objectives of SACRM can be summarized as two-fold.

(1) \textit{Participant Selection}.  Since the social attribute, task delay and reputation are crucial for the task crowdsourcing, SACRM should consider these factors and be able to select the well-suited participants for the sensing tasks and improve the crowdsourcing utility.

(2) \textit{Accurate Sensing Report Assessment and Adaptive Reward Allocation}. Since some malicious participants may submit bad sensing reports or contribute noting for their participated sensing tasks, SACRM should be able to accurately assess the submitted sensing reports and adaptively allocate the rewards to the participants based on the assessed report quality.

\section{The Proposed SACRM Scheme}
\label{sec4}
The proposed SACRM scheme consists of three components: (1) \textit{Participant Selection}, (2) \textit{Sensing Report Assessment and Adaptive Rewarding}, and (3) \textit{Reputation Management}. In \textit{Participant Selection}, we define a utility function to quantify the effect of social attributes, task delay and reputation on the crowdsourcing and formulate the participant selection problem as a combination optimization problem. Two participant selection algorithms are proposed to select the well-suited participants for the sensing task and maximize the crowdsourcing utility. In \textit{Sensing Report Assessment and Adaptive Rewarding}, we first evaluate the quality of submitted sensing reports in terms of the report veracity and report delay. And then, we propose a rewarding scheme to allocate task rewards based on the report assessment results. In \textit{Reputation Management},  the reputation of the participant is updated according to both of the assessed report quality and the cost performance ratio of the participant. The overview of the SACRM is described as~\ref{fig.schemeoverview}. To assist the understanding of the following paper, we summarized the frequently used notations in Table~\ref{table1}.
\begin{figure}[!t]
\includegraphics[width=0.5\textwidth]{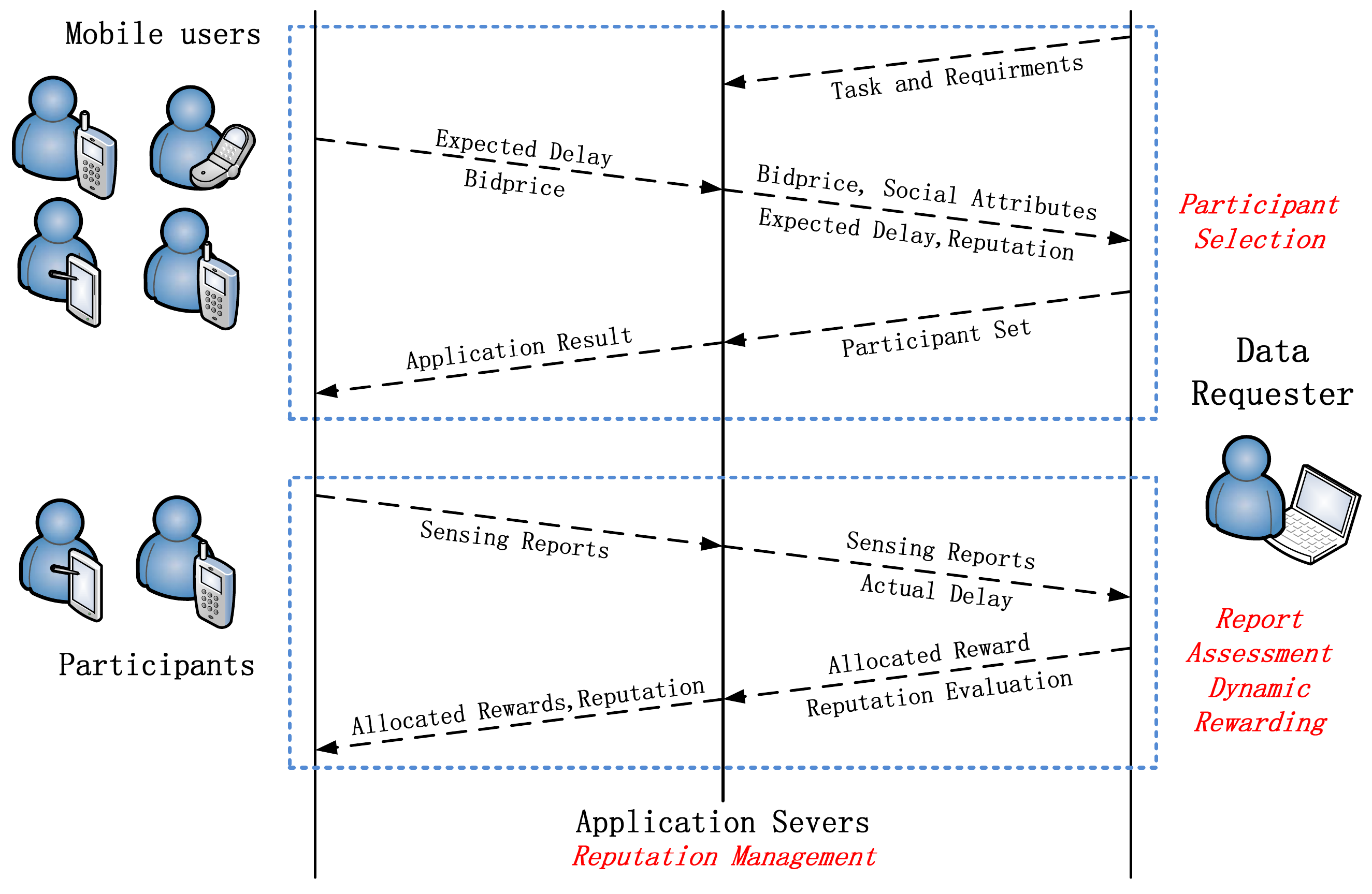}
\caption{Overview of the SACRM scheme}
\label{fig.schemeoverview}
\end{figure}

\begin{table}[!t]
    \caption{Frequently Used Notations}
    \centering
    \begin{tabular}{|l|l|}
         \hline
         Notation & Meaning \\
         \hline
         $\mathcal{U}$ & Set of mobile users $\{u_1, ..., u_n\}$ \\
         \hline
         $\mathcal{P}_t$ & Set of participants in task $t$, and $\mathcal{P}_t \subseteq  \mathcal{U}$\\
          \hline
         $B_t $ & Budget for finishing $t$\\
          \hline
         $d_t $ & Deadline of task $t$\\
          \hline
         $b_t^{i}$ & $u_i$'s bid price for $t$ \\
         \hline
         $d_t^{i}$ & $u_i$'s expected delay for finishing $t$\\
         \hline
         $e_t^{i}$ & Expected utility of choosing $u_i$ for $t$\\
         \hline
         $SA_{i}$ & Set of $u_i$'s social attributes \\
         \hline
         $TA_t$ & Set of social attributes interested by $t$\\
         \hline
         $R(u_i)$ & Reputation value of $u_i$\\
         \hline
         $Ie_t^{i}$ & Amplified $e_t^{i}$ to be non-negative integer \\
         \hline
         $sr_t^{i}$ & $u_i$'s sensing report for task $t$\\
         \hline
         $\Phi_t^{i}$ & Veracity score of $sr_t^i$\\
         \hline
         $\zeta_t^{i}$ & Delay deviation score of $sr_t^i$ \\
         \hline
         $v_t^{i}$ & Final report assessment score of $sr_t^i$\\
         \hline
         $tp_t^{i}$ & Allocated reward to $u_i$ for participating $t$\\
         \hline
         $r_t^{i}$ & $u_i$'s evaluated reputation score for participating $t$\\
         \hline
    \end{tabular}
     \label{table1}
\end{table}

\subsection{Participant Selection}
In SACRM, we consider three main factors, including social attributes, expected task delay and reputation, for participant selection. To make it clearer, we define the three factors as follows.
\begin{asparaenum}
\item \textit{\textbf{Social Attributes}}. Social attributes are the characteristics or features of an individual in his social life, such as interests, friend circle, living area. Generally, different tasks are interested in various social attributes and a large social attribute overlap between the task and the user indicates a potential matching and a high task quality.
\item \textit{\textbf{Expected Task Delay}}. The expected task delay for a specific task highly depends on the user's location and time availability. For each user $i$ and task $t$, the expected delay is defined as the expected duration from the time when $i$ is assigned the task $t$ until the time $i$ can finish this task. The expected delay indicates the timeliness of the crowdsourcing task, which is an important factor for participant selection, particularly in delay-sensitive tasks.
\item \textit{\textbf{Reputation}}. Due to the selfishness of the mobile users, a crucial part of the system is to assess if the quality and reliability of the reported sensed data deserve its bid price. We define the reputation of a mobile user $u_i$, denoted as $R(u_i)$, is a synthesized evaluation on the past sensing reports sent by $u_i$, as perceived by the platform. The platform maintains a reputation database to record the reputation value of each mobile user. When a new mobile user registers with the platform, the platform creates a unique ID and initializes an initial reputation value $R_0$ for the new user.
\end{asparaenum}

\subsubsection{Problem Formulation}
Participant selection is to select the well-suited mobile users to participate in the sensing task based on the task requirements and the application information, and hence to maximize the utility of the data requester. Therefore, we formulate the participant selection problem as follows.

A data requester $DR$ publishes a task $t$ and the task requirements of $t$ on the platform. The requirements of $t$ include the task budget $B_t$ for $t$, the interested social attributes set, denoted by $TA_t = \{ta_1, ta_2, ..., ta_k\}$, and the delay threshold $d_t$ that means the expected delay should be not larger than $d_t$. There is a set of mobile users, denoted by $\mathcal{U} = \{u_1, u_2, ..., u_n\} $, interested in participating in the task $t$. Here, $n \ge 2$. For each $u_i$, he has a set of personal social attributes, denoted by $SA_i =  \{sa_1^{i}, sa_2^{i}, ..., sa_{\tau_i}^{i}\}$, and a reputation value denoted as $R(u_i)$. In addition, $u_i$ estimates the expected delay $d_t^{i}$ to finish $t$ and the bid price $b_t^{i}$ for $t$, where $b_t^{i} \le B_t$. And then $u_i$ submits it to the platform, when $u_i$ is applying for the task $t$. We denote the expected utility of $DR$ as $e_t^{i}$ if he chooses the user $u_i$ to undertake the task $t$, and the total expected utility of $DR$ obtained from the crowdsourcing as $E_t$. According to the factor definitions, $e_t^{i}$ is determined by the overlap of $TA_t$ and $SA_i$, $d_t^{i}$ and $R(u_i)$, which will be detailed in the following section. Therefore, \textit{participant selection} is to choose a subset $\mathcal{P}_t$ of $\mathcal{U}$ to participate in the task $t$, to maximize the $E_t$. The problem is mathematically formulated as choosing $\mathcal{P}$ to
\begin{equation}
\label{eq.target}
\text{Maximize}~E_t = \sum_{i \in \mathcal{P}_t} e_t^{i};
\end{equation}
\begin{equation}
\label{eq.subject}
\begin{matrix}
\text{Subject~to}\begin{cases}& \sum_{i \in \mathcal{P}_t} b_t^{i} \le B_t;\\
& d_{t}^{i} \le d_t,~\text{for}~i \in \mathcal{P}_t; \\
& b_t^{i} \le B_t,~\text{for}~i \in \mathcal{P}_t.
\end{cases}
\end{matrix}
\end{equation}

\subsubsection{Utility Function}
We formulate the participant selection problem above, where the expected utility $e_t^{i}$ of $DR$ choosing $u_i$ to participate in the task $t$ is not accurately defined. In this section, we first define the utility functions of the social attributes, expected task delay and reputation to quantify the effect of them on participant selection. For a specific task $t$, the expected utility $e_t^{i}$ of choosing the user $u_i$ can be determined by the overlap of $SA_i$ and $TA_t$ , the expected delay $d_t^{i}$ and the reputation $R(u_i)$. Without loss of generality, we set $e_t^{i}$ consists of $f(SA_i, TA_t)$, $g(d_t^{i}, d_t)$ and $h(R(u_i))$. Here, $f(SA_i, TA_t)$, $g(d_t^{i}, d_t)$ and $h(R(u_i))$ denote the utility functions of social attributes overlap, expected delay and reputation value respectively. In order to facilitate the comparison of each utility $e_t^{i}$, we set $0 < e_t^{i} \le 1$.

The interested social attributes set of the task $t$ indicates the user with more common social attributes is expected to bring more benefits~\cite{35,36}, which is denoted by $TA_t = \{ta_1, ta_2, ..., ta_k\}$. Therefore, for each users $u_i$, the utility function of social attributes overlap $f(SA_i, TA_t)$ should be linearly dependent on the overlap ratio of $TA_t$ and $SA_i$. We define $f(SA_i, TA_t)$ as
\begin{align}
f(SA_i, TA_t) = (1 - \alpha) \dfrac{ \left |SA_i \cap TA_t\right | }{\left |TA_t\right |} + \alpha
\label{eq.socialfunction}
\end{align}
where $\left |TA_t\right |$ denotes the number of elements in $TA_t$; $ \left |SA_i \cap TA_t\right |$ denotes the number of common social attributes between $SA_i$ and $TA_t$; $\alpha$ denotes the default utility for the users without common social attributes with $TA_t$ and $0 < \alpha < 1$.

Expected task delay is another factor that should be considered in participant selection. For each user $u_i$, the expected delay varies with the different sensing tasks. If the expected delay $d_t^{i}$ of $u_i$ does not exceed the delay threshold $d_t$, $u_i$ is a participant candidate of the task $t$. Since a lower expected delay indicates a quicker sensing report, the delay utility function $g(d_t^{i}, d_t)$ should decrease with the increment of $d_t^{i}$. Thus, we define $g(d_t^{i}, d_t)$ as
\begin{align}
g(d_t^{i}, d_t) = (1 - \beta) (1 - \exp{(d_t^{i} - d_t)}) + \beta,~\text{if}~d_t^{i} \le d_t
\label{eq.delayfunction}
\end{align}
where $\beta$ denotes the default utility for the user with $d_t^{i} = d_t$ and $0 < \beta < 1$. The value of $g(d_t^{i}, d_t)$ equals to $\beta$ when $d_t^{i} = d_t$. While $g(d_t^{i}, d_t)$ approaches to 1 if $d_t^{i}$ is close to 0. The exponent function is adopted to stimulate a smaller expected delay, since $g(d_t^{i}, d_t)$ significantly decreases when $d_t^{i}$ is close to $d_t$.

Reputation is the last but absolutely not the least factor that is considered in participant selection. It indicates the quality of the sensing reports that the mobile user submitted in his past sensing tasks and the trustworthiness and cost performance ratio of the mobile user. Therefore, the reputation utility function $h(R(u_i))$ should be a monotonically increasing function. We set the maximum reputation value as $R_m$ and the minimum reputation value as $R_s$, then we have $R_s \le R(u_i) \le R_m$. Thus, the $h(R(u_i))$ is defined as
\begin{align}
h(R(u_i)) = \begin{cases}
\gamma + (1-\gamma)\ln\left ( 1+ \lambda \right );& \text{ if } R_0 \le R(u_i) \le R_m \\
\gamma \exp(R(u_i) - R_0);& \text{ if } R_s \le R(u_i) < R_0
\end{cases}
\label{eq.reputationfunction}
\end{align}
where $\lambda = \dfrac{(e-1)(R(u_i)-R_0)}{R_m-R_0}$; $R_0$ is an initial reputation value for a new mobile user; $\gamma $ denotes the default utility for the user with $R(u_i) = R_0$ and $0 < \gamma < 1$. The $\gamma$ can be set as a neutral value $0.5$, since we have $0 < h(R(u_i)) \le 1$. The exponent function makes the reputation value decrease sharply if $R(u_i) < R_0$, and the logarithm function markedly increases the reputation value if $R(u_i) \ge R_0$.

Combining the Eq.~(\ref{eq.socialfunction}), (\ref{eq.delayfunction}) and (\ref{eq.reputationfunction}), we define the utility function $e_t^{i}$ of choosing the user $u_i$ as
\begin{align}
e_t^{i} = w_s \cdot f(SA_i, TA_t) + w_d \cdot g(d_t^{i}, d_t) + w_r \cdot h(R(u_i))
\label{eq.delayfunction}
\end{align}
where $w_s$, $w_d$ and $w_r$ denote the weights of social attributes, delay and reputation respectively; $0 \le w_s, w_d, w_r \le 1$ and $w_s + w_d + w_r = 1$.

For different mobile sensing tasks, task requirements might be various and hence the criteria of participant selection vary in each task too. For instance, expected task delay should be the dominating factor in participant selection of the delay-sensitive sensing tasks, such as ``\textit{Take a photo for the Davis Centre building of the University of Waterloo in 5 minutes}". So, we can increase the weight of the expected delay utility, i.e., $w_d$ and set smaller values to $w_s$ and $w_r$. Correspondingly, $w_s$ should be increased for the speciality-sensitive sensing tasks and $w_r$ should be improved for the verasity-sensitive tasks. In summary, we can dynamically adjust the values of the three weights to fit for the various task requirements.

Note that, the values of $f(SA_i, TA_t)$, $g(d_t^{i}, d_t)$ and $h(R(u_i))$ are in $(0, 1]$. It means the value of $e_t^{i}$ is also in $(0, 1]$, which is useful for the utility comparison in participant selection.

\subsubsection{Proposed Participant Selection Algorithms}
Since the utility function $e_t^{i}$ has been determined in the previous section, the objective function of participant selection can be rewritten as
\begin{align}
\label{eq.target}
E_t = \sum_{i \in \mathcal{P}} \left ( w_s \cdot f(SA_i, TA_t) + w_d \cdot g(d_t^{i}, d_t) + w_r \cdot h(R(u_i)) \right )
\end{align}

In this subsection, we describe the participant selection algorithm in detail. We first prove that finding an optimal participant set for the task $t$ is an NP-hard problem (i.e., can be reduced to the \textit{0-1 Knapsack Problem}~\cite{add_knapsack}).
\begin{thm}
\label{thm.nphard}
The participant selection algorithm is NP-hard.
\end{thm}

\begin{proof}
We aim to reduce our problem to the \textit{0-1 Knapsack Problem}: Give $n$ items $\{z_1, z_2, ..., z_n\}$ where $z_i$ has a value $v_i | v_i \ge 0$ and weight $w_i | w_i \ge 0$. The maximum weight that we can carry in the bag is $W$. $x_i = 0~or~1$ denotes if the item $z_i$ should be put into the bag. The \textit{0-1 Knapsack Problem} is to determine $\{x_1, x_2, ..., x_n\}$ to
\begin{align*}
\text{maximize~~}V = \sum_{i = 1}^{n} v_i x_i,\\
\text{subject~to}~\sum_{i = 1}^{n} w_i x_i \le W.
\end{align*}
Then we construct our participant selection problem as follows. Denote $\mathcal{P}_1 = \{p_1, p_2, ..., p_{n_1}\}$ as the set of participant candidates excluding the ones with the expected delay larger than $d_t$ or the bid price larger than $B_t$. We use $x_j = 0~or~1$ to denote if the candidate $p_j$ should be chosen to participate in the task $t$. Then the participant selection problem changes to determine $\{x_1, x_2, ..., x_{n_1}\}$ to
\begin{align*}
\text{maximize~~}E_t = \sum_{j = 1}^{n_1} e_t^{j} x_j,\\
\text{subject~to}~\sum_{j = 1}^{n_1} b_t^{j} x_j \le B_t.
\end{align*}

Therefore, the \textit{0-1 Knapsack Problem} is successfully reduced to the participant selection problem, which finishes the proof.
\end{proof}

Since the participant selection problem is proved as an NP-hard problem, the optimal solution can not be achieved by a polynomial time algorithm, but it can be obtained by a pseudo-polynomial time algorithm. We first make some modifications to our problem. Let $\mathcal{P}_1 = \{p_1, p_2, ..., p_{n_1}\}$ as the set of participant candidates. We map the utility set $\{e_t^1, e_t^2, ..., e_t^{n_1}\}$ into non-negative integers $\{Ie_t^1, Ie_t^2, ..., Ie_t^{n_1}\}$ by multiplying each of them by a amplification factor $\delta$. Define $Ie_t^{max}$ as the maximum value in the amplified utility set. For each $i \in \{1, ..., n\}$ and $k \in \{1, ..., \sum{Ie_t^i}\}$, we define $A[i, k] = \min \left \{ \sum_{j=1}^{i}b_t^{j}~|~\sum_{j=1}^{i}Ie_t^{j} = k\right \}$ and
\begin{asparaitem}
\item $A[i, k]$ is subset of $\mathcal{P}_1$ whose total utility is exactly $k$ and whose total payment is minimized;
\item $A[0, 0] = 0$ and for each $k \in \{1, ..., \sum{Ie_t^i}\}$, we have $A[0, k] = B_t + 1$.
\end{asparaitem}

Then, we can recursively calculate the $A[i+1, k]$ as
\begin{equation*}
A[i+1, k] = \begin{cases}
A[i, k], & \text{ if } Ie_t^{i+1} > k;\\
\min \{A[i, k], A[i, k-Ie_t^{i+1}]+b_t^{i}\}, & \text{ if } Ie_t^{i+1} \le k.
\end{cases}
\end{equation*}

Therefore, the optimal utility is $\max\{k~|~A[n_1, k] \le B_t\}$. We describe the pseudo code for the dynamic programming participant selection algorithm as Alg.~\ref{alg.1}.
\begin{algorithm}[!t]
\LinesNumbered
\DontPrintSemicolon
\SetKwData{Left}{left}\SetKwData{This}{this}\SetKwData{Up}{up}
\SetKwFunction{Union}{Union}\SetKwFunction{FindCompress}{FindCompress}
\SetKwInOut{Input}{input}\SetKwInOut{Output}{output}
	\Input{The participant candidate set $\mathcal{P}_1[1,...,n]$, the bid price set $b[1,...,n]$, and $B_t$, $TA_t$, $d_t$, $\{SA[1], ..., SA[n]\}$, $\{d_t[1], ..., d_t[1]\}$, $\{R[1], ..., R[n]\}$;}
	\Output{The selected participant set $P$, the optimal utility value $opt$;}
	\BlankLine
	\For{$i$ from $1$ to $n$}{
	         $e_t[i] \leftarrow w_s \cdot f(SA[i], TA_t) + w_d \cdot g(d_t[i], d_t) + w_r \cdot h(R[i])$;\;
		$Ie_t[i] \leftarrow e_t[i] \cdot \delta$;\;
	}
	$A[0, 0] \leftarrow 0$;\;
	\For{$k$ from $0$ to $B_t$}{
		$A[0, k] \leftarrow B_t + 1$;\;
	}
	\For{$i$ from $1$ to $n$}{
		\For{$j$ from $0$ to $\sum{Ie_t[i]}$}{
			\eIf{$Ie_t[i] \le j$ \&\& $A[i-1, j-Ie_t[i]] + b[i] < A[i-1, j]$}{
				$A[i, j] \leftarrow A[i-1, j-Ie_t[i]] + b[i]$;\;
				$z[i, j] \leftarrow 1$;\;
			}{
				$A[i, j] \leftarrow  A[i-1, j]$;\;
				$z[i, j] \leftarrow 0$;\;
			}
		}
	}
	$opt \leftarrow 0$;\;
	\For{$j$ from $0$ to $\sum{Ie_t[i]}$}{
		\If{$opt \le j$ \&\& $A[n, j] \le B_t$}{
			$opt \leftarrow j / \delta$;\;
		}
	}
	$BB \leftarrow opt \cdot \delta$;\;
	\For{$i$ from $n$ downto $1$}{
		\If{$z[i,BB] == 1$}{
			$P \leftarrow P + \mathcal{P}_1[i]$;\;
			$BB \leftarrow BB - b_t[i]$;
		}
	}
	\Return $P$ and $opt$;
	\caption{The Dynamic Programming Participant Selection Algorithm}
	\label{alg.1}
\end{algorithm}

The time complexity and space complexity of Alg.~\ref{alg.1} are $O(n\sum{Ie_t[i]}) \le O(n \cdot nIe_t^{max})$ and $O(n\sum{Ie_t[i]})$, where $Ie_t^{max}$ is the maximum value in $Ie_t[0, ..., n]$. However, if we use 1-dimensional array $A[0, ..., \sum{Ie_t[i]}]$ to store the current optimal values and pass over this array $i+1$ time, recalculating from $A[\sum{Ie_t[i]}]$ to $A[0]$ every time, we can obtain the optimal value for only $O(\sum{Ie_t[i]})$ space. However, Alg.~\ref{alg.1} is not a fully polynomial algorithm for participant selection. Since $\sum{Ie_t[i]}$ is not polynomial in the length of the input of the problem, we consider the algorithm is efficient only if $\sum{Ie_t[i]}$ is small or polynomial in $n$.   

To reduce the time complexity of Alg.\ref{alg.1}, we propose a \textit{fully polynomial time approximation scheme} (FPTAS)~\cite{add_FPATS} to select participants for the task crowdsourcing. The basic idea of the FPTAS is to ignore a certain number of least significant bits of the utility, depending on the error parameter $\epsilon $. Such that, the modified utilities can be viewed as numbers bounded by a polynomial in $n$ and $1/\epsilon$. The pseudo code for the FPTAS is described as Alg.~\ref{alg.2}.
\begin{algorithm}[!t]
\LinesNumbered
\DontPrintSemicolon
\SetKwData{Left}{left}\SetKwData{This}{this}\SetKwData{Up}{up}
\SetKwFunction{Union}{Union}\SetKwFunction{FindCompress}{FindCompress}
\SetKwInOut{Input}{input}\SetKwInOut{Output}{output}
	\Input{The approximation error $\epsilon > 0$, and the same inputs as Alg.~\ref{alg.1};}
	\Output{The selected participant set $P^{'}$, the approximated optimal utility value $s$;}
	\BlankLine	
	\For{$i$ from $1$ to $n$}{
		$Ie_t[i] \leftarrow e_t[i] \cdot \delta$;\;
	}
	Find the maximum value $Ie_t^{max}$ from $Ie_t[1, ..., n]$;\;
	$Q \leftarrow \dfrac{\epsilon \cdot Ie_t^{max}}{n}$;\;
	\For{$i$ from $1$ to $n$}{
		$Ie_t^{'}[i] \leftarrow \left \lfloor \dfrac{Ie_t[i]}{Q} \right \rfloor$;\;
	}
	Compute the selected participant set $P^{'}$ and the approximated optimal utility $s$ with Alg.~\ref{alg.1} using the $Ie_t^{'}[1, ..., n]$ as the utility input amplified to the integers;\;
	\Return $P^{'}$ and $s$;\;
	\caption{participant selection - FPTAS}
	\label{alg.2}
\end{algorithm}
	
\subsection{Report Assessment and Rewarding Scheme}
According to the participant selection scheme, the optimal expected utility can be achieved from the task crowdsourcing. However, due to the selfishness of the participants and uncertainties, the quality of the sensing reports should be evaluated to determine the trustworthiness and value of the sensing reports. There has been a large number of research efforts on the data quality assessment in the field of data mining~\cite{31,32}. In SACRM, we particularly focus on the two metrics, the veracity and actual delay of the sensing reports.

Generally, a sensing task is outsourced to multiple participants in mobile sensing to ensure the veracity of sensing reports. Similar sensing reports are mutually supportive to each other, while conflicting or inconsistent reports compromise the veracity of each other. Therefore, we can evaluate the veracity of the sensing report based on the amount of supports and conflicts it obtains from other sensing reports. We group all the reports for a specific sensing task $t$ into a collection $C_t$ and measure the data similarity for each report based on the similarity function.

Assume that the similarity score $S(sr_t^{i}, sr_t^{j})$ of any two sensing reports $sr_t^{i}$ and $sr_t^{j}$ in $C_t$ ranges from -1 to 1~\cite{1,31,32}, where -1 means completely conflicting and 1 means exactly consistent. Notably, since the similarity function design has been widely studied~\cite{1,33,34}, our focus is on how to utilize the similarity scores determined by the similarity function to evaluate the veracity of the report. We define the report veracity assessment $\Phi_t^{i}$ of each sensing report $sr_t^{i}$ in $C_t$ as
\begin{equation*}
\Phi_t^{i} = \dfrac{1 + \sum_{i, j \in C_t, i \neq j}S(sr_t^{i}, sr_t^{j}) \cdot e^{-\frac{1}{|C_t|}}}{2 \cdot (|C_t| -1)},
\end{equation*}
where $|C_t|$ is the number of sensing reports in the collection $C_t$. The $e^{-\frac{1}{|C_t|}}$ indicates the influence of the similarity score is reduced with the decrease of the number of the sensing reports.

The actual report delay is another metric for the report quality assessment. Since the expected delay is considered in participant selection as a crucial factor, the deviation between the actual report delay and the expected delay should be introduced into the report assessment to evaluate the timeliness. If the actual delay of $u_i$ is much larger than the expected delay, it would cause a negative impact on the report assessment. Therefore, if we denote $ad_t^{i}$ as the actual delay of the participant $i$ for the task $t$, the delay deviation assessment $\zeta_t^{i}$ can be defined as
\begin{equation*}
\zeta_t^{i} = \begin{cases}
~~1, &\text{if}~ad_t^{i} \le d_t^{i}+\sigma_t; \\
1 - \vartheta \cdot (1-e^{\left ( \frac{d_t^{i}+\sigma_t-ad_t^{i}}{d_t-d_t^{i}-\sigma_t} \right ) \cdot \varphi_1}), &\text{if}~d_t^{i}+\sigma_t< ad_t^{i} \le d_t.
\end{cases}
\end{equation*}
where $\sigma_t$ is a delay adjustment factor for the task $t$; $\varphi_1$ is an amplification factor to amplify the effect of $\left ( \frac{d_t^{i}+\sigma_t-ad_t^{i}}{d_t-d_t^{i}-\sigma_t} \right ) $ on the delay assessment, for instance, we can set $\varphi_1 = 5$;$1-\vartheta(1-1/e)$ as the lower bound of the delay deviation if $ad_t^{i} \le d_t$, and $0 \le \sigma_t \le d_t - \max\{d_t^{i}\}$, $0 < \vartheta \le 1$. Since $\vartheta$ can be any real number between 0 and 1, the delay assessment function can be adaptive for different application requirements by adjusting $\vartheta$. Meanwhile, we consider the delay assessment score should decrease with the increment of the delay deviation between the actual delay and the expected delay. Therefore, we use $1 - \vartheta \cdot (1-e^{\left ( \frac{d_t^{i}+\sigma_t-ad_t^{i}}{d_t-d_t^{i}-\sigma_t} \right ) \cdot \varphi_1})$ to rate the delay deviation score. Although the situation of $ad_t^{i} > d_t$ is not defined here, the sensing reports with such situation should be identified as invalid for the task.

Based on the discussed two metrics above, we can integrate them into a final report assessment function $v_t^{i}$ as follows.
\begin{equation}
v_t^{i} =\mathbf{\Gamma}(SR_t, ad_t^{i}, d_t^{i}, d_t)= \begin{cases}
w_x \cdot \Phi_t^{i} + (1-w_x) \cdot \zeta_t^{i}, &\text{if}~ad_t^{i} \le d_t; \\
~~0, &\text{if}~ad_t^{i} > d_t. \\
\end{cases}
\label{eq.assessment}
\end{equation}
where $w_x$ is the weight of report veracity in the report assessment and $0 \le w_x \le 1$. Note that, since both of the $\Phi_t^{i}$ and $\zeta_t^{i}$ range from 0 to 1, the $v_t^{i}$ should range from 0 to 1 too. With the weighted integration of report veracity and delay deviation, the report assessment can be adjusted to fit for the various task requirements.

With the report assessment defined above, we dynamically allocate the rewards to the participants. Denote the assessment result of the sensing report $sr_t^{i}$ as $v_t^{i}$. We define the reward allocation function $tp_t^{i}$ as
\begin{equation}
tp_t^{j} = \mathbf{\Delta}(v_t^{i}, b_t^{i}) =\begin{cases}
b_t^{i} \cdot e^{(v_t^{i}-v_t(h))\cdot \varphi_2}, &\text{if}~v_t^{i} < v_t(h); \\
~~b_t^{i}, &\text{otherwise}. \\
\end{cases}
\label{eq.rewards}
\end{equation}
where the $v_t(h)$ is the threshold of the report assessment result determined by the task requirement, and $0 < v_t(h) < 1$; $\varphi_2$ is an amplification factor to amplify the effect of the $v_t^{i}-v_t(h)$ on the reward allocation, for instance, we can set $\varphi_2 = 2$. If the report assessment value is not less than $v_t(h)$, the sensing report is identified as a good sensing report and the reward for $u_i$ should be his bid price; otherwise, the sensing report would be identified as a poor sensing report and the reward for the participant $i$ would be reduced with the decrease of the $v_t^{i}$.

Based on the Eq.~(\ref{eq.assessment}), (\ref{eq.rewards}), we describe the report assessment and rewarding scheme by Alg.~\ref{alg.3}.
\begin{algorithm}[!t]
\LinesNumbered
\DontPrintSemicolon
\SetKwData{Left}{left}\SetKwData{This}{this}\SetKwData{Up}{up}
\SetKwFunction{Union}{Union}\SetKwFunction{FindCompress}{FindCompress}
\SetKwInOut{Input}{input}\SetKwInOut{Output}{output}
	\Input{The sensing report set $\{sr_t[1], ..., sr_t[n]\}$, the actual delay set $\{ad_t[1], ..., ad_t[n]\}$ and other inputs same as Alg.~\ref{alg.1};}
	\Output{The report assessment result set $\{v_t[1], ..., v_t[n]\}$, and the allocated reward set $\{rp_t[1], ..., rp_t[n]\}$;}
	\BlankLine	
	\For{$i$ from $1$ to $n$}{
		$v_t[i] \leftarrow \mathbf{\Gamma}(SR_t, ad_t^{i}, d_t^{i}, d_t)$;\;
		$rp_t[i] \leftarrow \mathbf{\Delta}(v_t^{i}, b_t^{i})$;\;
		Allocate the reward $rp_t[i]$ to $i$;\;
	}
	\Return $\{v_t[1], ..., v_t[n]\}$ and $\{rp_t[1], ..., rp_t[n]\}$;\;
	\caption{Report Assessment and Rewarding Scheme}
	\label{alg.3}
\end{algorithm}

\subsection{Reputation Management}
Due to the selfishness of individuals, participants are eager to obtain more benefits with fewer efforts in mobile crowdsensing. Furthermore, there might be some malicious mobile users, who maliciously participate a number of sensing tasks and submit bad sensing reports with a high probability to jeopardize the mobile crowdsensing system. Although the report assessment and rewarding scheme can economically punish poor sensing reports and stimulate the improvement of report quality, it is still crucial to establish a reputation system to provide a synthesized evaluation on the the past sensing reports sent by each participant. Different from the existing reputation systems which only focus on the trustworthiness evaluation of participants, we introduce the bid price into the reputation system to evaluate the cost performance ratio of participants. The participant finishing the sensing task with the same quality but with a lower bid price should be rated a higher reputation by the data requester. Consequently, the crowdsourcing cost and the quality of sensing reports would be improved by the double stimulation (i.e., the trustworthiness and bid price evaluation in our reputation system).

Denote the participant set by $P$, and $v_t^{i}$ as the assessment result of the sensing report $sr_t^{i}$, and $b_t^{i}$ as the bid price of the participant $u_i$. Therefore, we define the reputation evaluation function $r_t^{i}$ as
\begin{equation}
r_t^{i} = \mathbf{\Lambda}(v_t^{i}, b_t^{i}) =\begin{cases}
\kappa \cdot (1-e^{\frac{-v_t^{i}/\sum_{i \in P}v_t^{i}}{b_t^{i}/\sum_{i \in P}b_t^{i}}}), &\text{if}~v_t^{i} \ge v_t(h); \\
~~-\eta, &\text{otherwise}. \\
\end{cases}
\label{eq.reputation}
\end{equation}
where $\kappa$ is a reward factor and $\eta$ is a punishment factor. To stimulate a better report quality, the reputation evaluation function should be defined asymmetrically, which means $\eta \gg \kappa$. We adopt a factor $1-e^{\frac{v_t^{i}/\sum_{i \in P}v_t^{i}}{b_t^{i}/\sum_{i \in P}b_t^{i}}}$
 to simulate a higher cost performance ratio.

Denote the final reputation of the participant $u_i$ in the platform is $R(u_i)$, and the maximum and minimum reputation values are $R_m$ and $R_s$. Then, we can integrate the reputation as
\begin{equation}
R(u_i) = \begin{cases}
R_s , &\text{if}~R(i) + r_t^{i}  < R_s; \\
R_m , &\text{if}~R(i) + r_t^{i}  > R_m; \\
R(u_i) + r_t^{i} , &\text{otherwise}. \\
\end{cases}
\label{eq.integration}
\end{equation}

Based on the Eq.~(\ref{eq.reputation}) and (\ref{eq.integration}), we describe the reputation management scheme by Alg.~\ref{alg.4}.
\begin{algorithm}[!t]
\LinesNumbered
\DontPrintSemicolon
\SetKwData{Left}{left}\SetKwData{This}{this}\SetKwData{Up}{up}
\SetKwFunction{Union}{Union}\SetKwFunction{FindCompress}{FindCompress}
\SetKwInOut{Input}{input}\SetKwInOut{Output}{output}
	\Input{The report assessment result set $\{v_t[1], ..., v_t[n]\}$ and the bid price set $\{b_t[1], ..., b_t[n]\}$;}
	\Output{The reputation set $\{r_t[1], ..., r_t[n]\}$ for the task $t$;}
	\BlankLine	
	\For{$i$ from $1$ to $n$}{
		$r_t[i] \leftarrow \mathbf{\Lambda}(v_t^{i}, b_t^{i})$;\;
		Integrate the $r_t[i]$ into $u_i$'s final reputation $R(u_i)$ based on Eq.~(\ref{eq.integration});\;
	}
	\Return $\{r_t[1], ..., r_t[n]\}$;\;
	\caption{Reputation Management Scheme}
	\label{alg.4}
\end{algorithm}

\section{Performance Analysis}
\label{sec5}
In this section, we theoretically analyze the performance of the proposed participant selection algorithms.

\begin{lem}
The time complexity and space complexity of Alg.~\ref{alg.2} are both $O(n^3\dfrac{1}{\epsilon})$.
\label{lem1}
\end{lem}
\begin{proof}
According to Alg.~\ref{alg.1}, the time complexity of dynamic programming is $O(n\sum{Ie_t[i]}) \le O(n \cdot nIe_t^{max})$.

In Alg.~\ref{alg.2}, we divide each utility by the factor $Q = \dfrac{\epsilon \cdot Ie_t^{max}}{n}$ to a new utility $Ie_t^{'}[i]$. Therefore, the time complexity of Alg.~\ref{alg.1} changes to $O(n \cdot nIe_t^{max'})$, where $Ie_t^{max'} = \left \lfloor \dfrac{Ie_t^{max}}{Q} \right \rfloor$. It means $O(n \cdot nIe_t^{max'}) = O(n \cdot n\left \lfloor \dfrac{Ie_t^{max}}{Q} \right \rfloor =  O(n^{2} \left \lfloor \dfrac{n}{\epsilon} \right \rfloor \le O(n^3\dfrac{1}{\epsilon})$.

Similarly, we can prove the space complexity of Alg.~\ref{alg.2} changes to $O(n^3\dfrac{1}{\epsilon})$.
\end{proof}

\begin{lem}
Denote $s$ as the output utility of Alg.~\ref{alg.2}, and $opt$ as the optimal utility. Then, we have $s \ge (1-\epsilon) opt$.
\label{lem2}
\end{lem}
\begin{proof}
Let $P$ denote the optimal selected participant set, $P^{'}$ denote the selected participant set of Alg.~\ref{alg.2}. For any amplified utility $Ie_t[i]$, because of the rounding down, we have $Ie_t[i] / Q - Ie_t^{'}[i] \le 1$, where $Q = \dfrac{\epsilon \cdot Ie_t^{max}}{n}$. Therefore,
\begin{align*}
\sum_{i \in P}le_t[i] - Q \cdot \sum_{i \in P}Ie_t^{'}[i] \le nQ.
\end{align*}

The dynamic programming steps return the optimal selected participant set with the new utility set $\{Ie_t^{'}[1], ..., Ie_t^{'}[n]\}$. Therefore,
\begin{align*}
\sum_{i \in P^{'}}Ie_t[i] \ge \sum_{i \in P}Ie_t^{'}[i] \ge \sum_{i \in P}Ie_t[i] - nQ = \sum_{i \in P}le_t[i] - \epsilon Ie_t^{max}
\end{align*}

Since for each $i$ we have $b_t^{i} \le B_t$, the optimal utility should be not less than $Ie_t^{max} $, i.e., $\sum_{i \in P}le_t[i] \ge Ie_t^{max}$. Therefore,
\begin{align*}
\sum_{i \in P^{'}}le_t[i] \ge \sum_{i \in P}le_t[i] - \epsilon Ie_t^{max} \ge (1-\epsilon)\sum_{i \in P}le_t[i].
\end{align*}

Since $le_t[i] = e_t[i] \cdot \delta$ and $\sum_{i \in P^{'}}le_t[i] = s \cdot \delta$ and $\sum_{i \in P}le_t[i] = opt \cdot \delta$, we have $s \ge (1-\epsilon) opt$.
\end{proof}

\begin{thm}
Alg.~\ref{alg.2} is a fully polynomial time approximation scheme (FPTAS) for participant selection.
\label{thm2}
\end{thm}
\begin{proof}
An FPTAS is an algorithm that takes an instance of an optimization problem and a parameter $\epsilon > 0$ and in polynomial time (both polynomial in $n$ and $1/\epsilon$) produces a solution that is within a factor $(1 + \epsilon)$ of being optimal or ($(1 - \epsilon)$ for maximization problems). According to the Lemma~\ref{lem1} and Lemma~\ref{lem2}, we can prove Alg.~\ref{alg.2} is a FPTAS for participant selection.
\end{proof}

\section{Performance Evaluation}
\label{sec6}
We evaluate the performance of SACRM based on Java simulations. In our simulations, we setup 20 mobile users who are interested in the published sensing task. The sensing task has 10 interested social attributes and each mobile user has 10 social attributes. Therefore, the number of overlapping social attributes ranges from 0 to 10. Each mobile user submits a expected delay to the data requester. The expected delay ranges from 1 to 45. And the actual task delay follows a normal distribution where the expected value is the expected task delay.  The reputation of each user is randomly assigned, which ranges from 0.1 to 1. If the reputation of a mobile user $rp$ is below 0.3 is a dishonest user who submits bad sensing reports with a high probability $(1-rp)$. And the user with reputation $rp$ higher than 0.3 would submit bad sensing reports with a probability of $0.4 \times (1-rp)$. The submitted report is assessed as a quality score by our report assessment scheme. We use the sum of the assessed quality scores as the crowdsourcing utility, denoted as $utility$.

\subsection{Participant Selection Scheme Evaluation}
In SACRM, we consider social attributes, expected delay and reputation in our participant selection scheme. In this section, we evaluate the effect of each factor on crowdsourcing by comparing our SACRM scheme with the Greedy Algorithm (GA) where the participant with the lowest bid price has the priority to be selected. In order to highlight the effect of each factor on participant selection, the weights of the other two factors are set to very low values (e.g., 0.05) to play a minor role in the evaluation.

Fig.~\ref{fig.social} shows the crowdsourcing utility comparison between SACRM and GA. It is shown that both of the utility of the SACRM and GA increase with the increasing task budget. It is obvious that a higher task budget can recruit more participants and create a higher utility. However, since GA does not consider the social attributes, the utility of the SACRM is significantly higher than the utility of the GA under each task budget. It demonstrates that social attributes play a significant role in participant selection and SACRM greatly improves the crowdsourcing utility. Fig.~\ref{fig.delay} shows the total actual delay comparison between SACRM and GA. The total actual delay in this figure means the sum of the actual delays of the submitted sensing reports. Similar with the Fig.~\ref{fig.social}, a higher task budget indicates more participants and higher total actual delay. Nevertheless, the SACRM has a remarkably lower total actual delay than GA, which indicates an enhanced crowdsourcing utility. The effect of the reputation is evaluated in Fig.~\ref{fig.reputation}. In order to evaluate the effect of the distribution of reputation values on SACRM, we compare the performance of SACRM and GA under the random distribution (RD) and normal distribution (ND). It can be seen that the SACRM considering the reputation in participant selection brings a significantly higher crowdsourcing utility than the GA, both in RD and ND. Meanwhile, since malicious mobile users (i.e., the mobile users with low reputation values) are less in ND than in RD, the utility of SACRM in ND is lightly higher than in RD. This simulation result also proves that the distribution of reputation values has little impact on the performance of SACRM. Therefore, combining Fig.~\ref{fig.social}, \ref{fig.delay}, \ref{fig.reputation}, it can be demonstrated that social attributes, expected delay and reputation are necessary to be considered in participant selection and SACRM leads to a significantly improved crowdsourcing utility.
\begin{figure}[!t]
\includegraphics[width=0.42\textwidth]{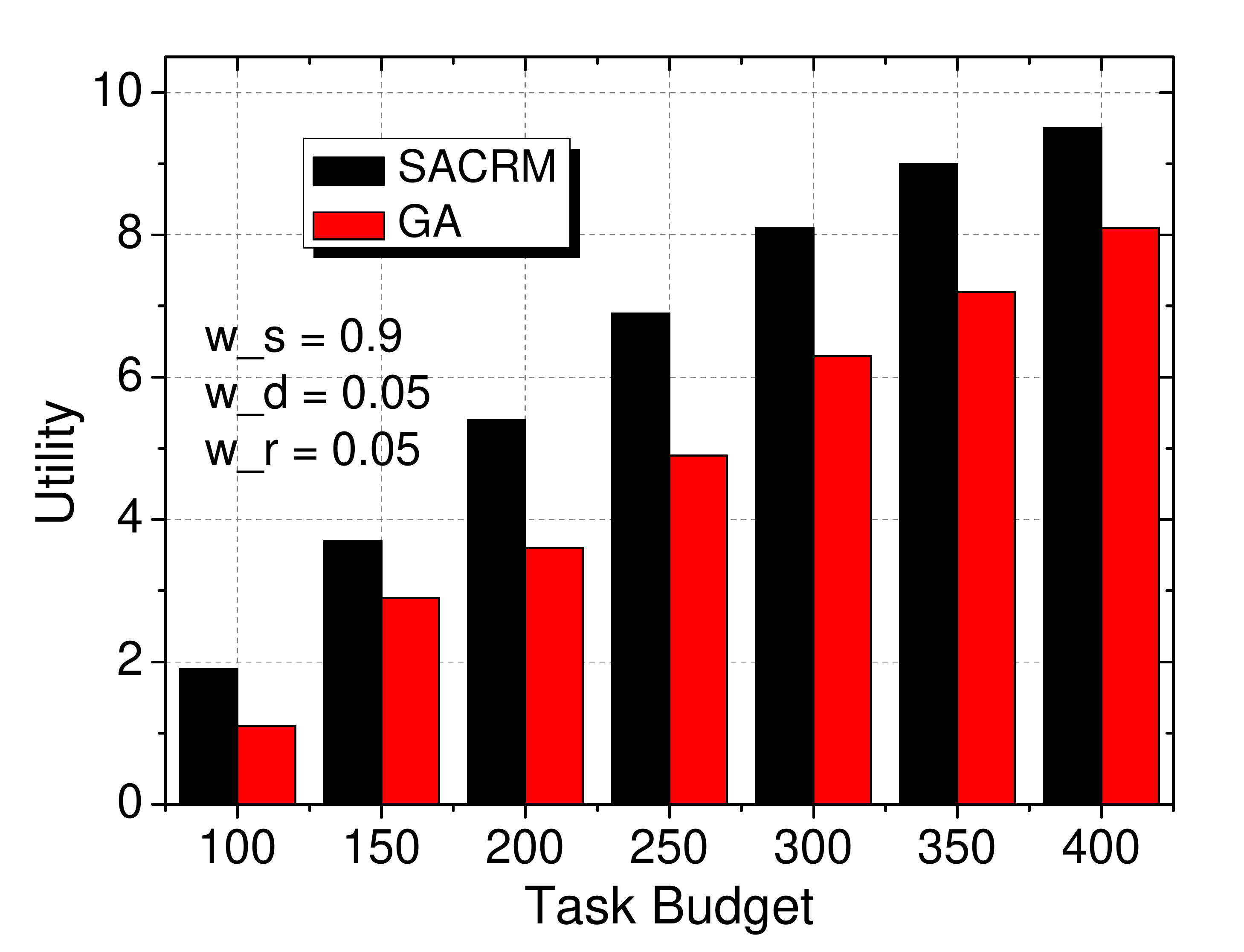}
\caption{The Effect of Social Attributes on Crowdsourcing.}
\label{fig.social}
\end{figure}
\begin{figure}[!t]
\includegraphics[width=0.42\textwidth]{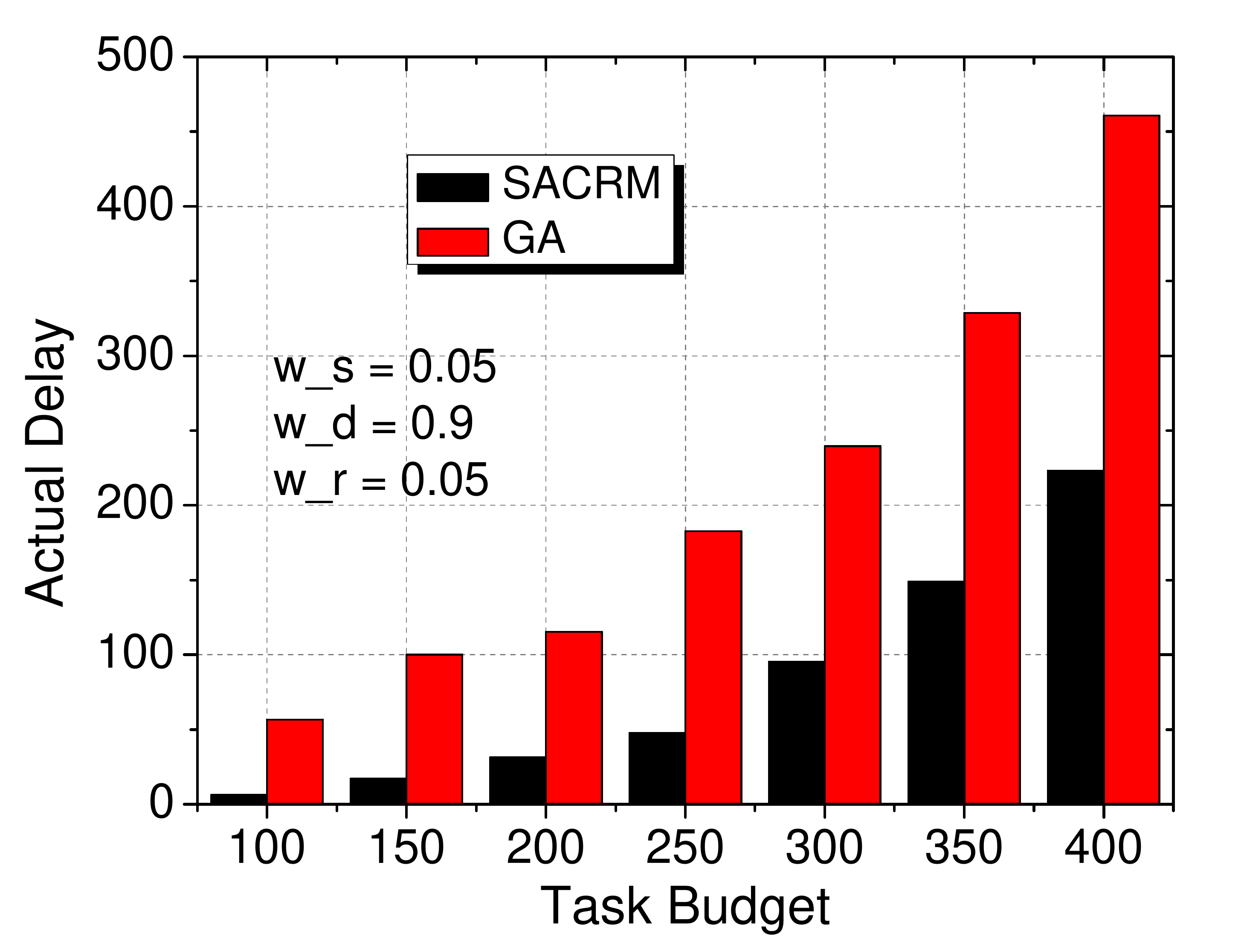}
\caption{The Effect of Expected Delay on Crowdsourcing.}
\label{fig.delay}
\end{figure}
\begin{figure}[!t]
\includegraphics[width=0.42\textwidth]{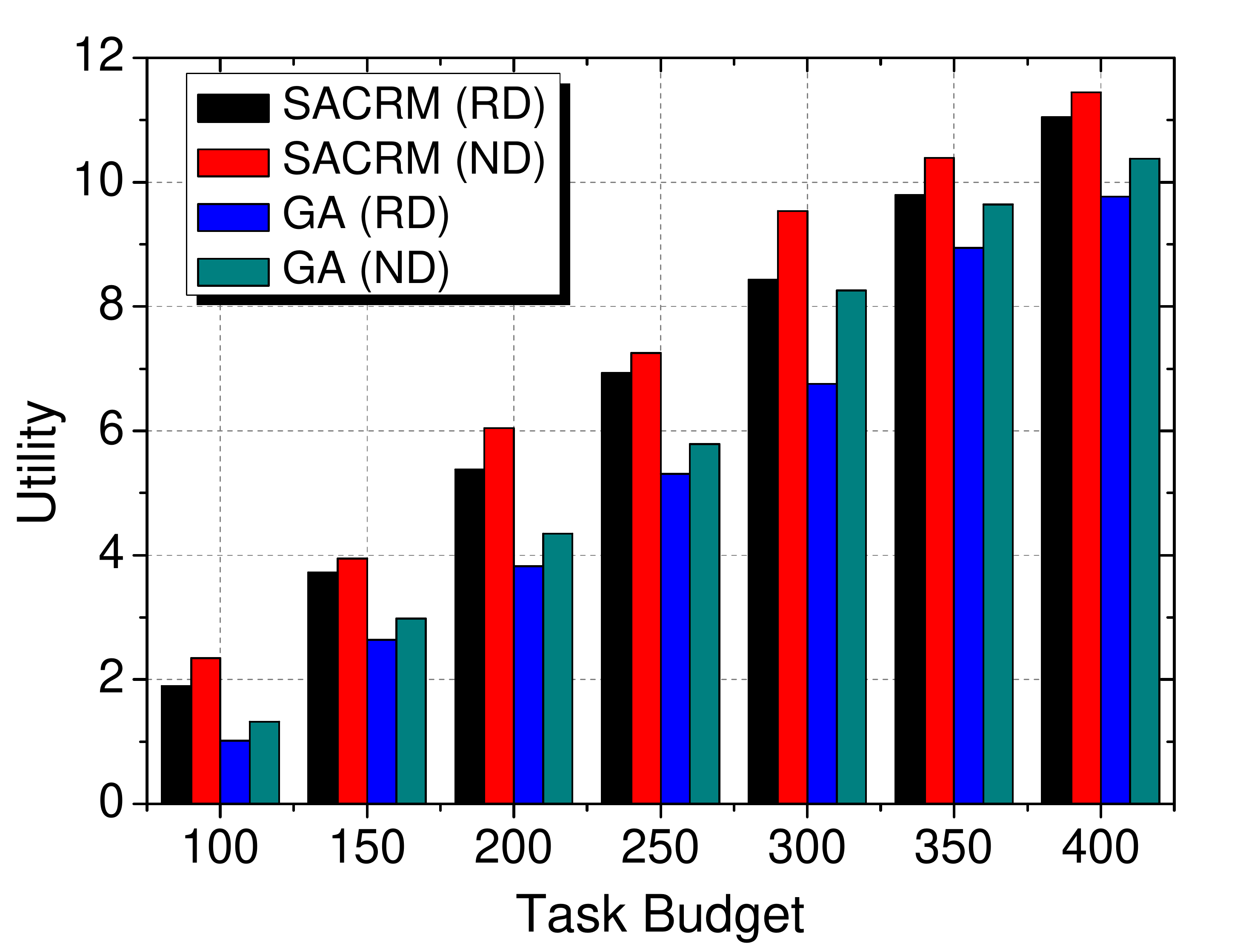}
\caption{The Effect of Reputation on Crowdsourcing\protect\footnotemark[2].}
\label{fig.reputation}
\end{figure}
\footnotetext[2]{RD means the reputation values of mobile users are followed by random distribution, while ND means the reputation values of mobile users are followed by normal distribution.}

Fig.~\ref{fig.scalability} shows the performance comparison between SACRM and GA under the increasing mobile users. As shown in this figure, when the task budget is fixed, SACRM can produce an increasing utility with a larger number of mobile users, while the utility of GA fluctuates with the increment of mobile users. It indicates that SACRM is scalable and can always achieve an optimized utility by selecting the well-suited participants from a set of mobile users. With the increment of mobile users, we can have more choices to find the best-suited participants for our task. That is why the expected utility increases with the increasing mobile users. However, GA always chooses the mobile user with lowest bid price for the task without considering the underlying relationship, which causes a low and fluctuating utility curve.
\begin{figure}[!t]
\includegraphics[width=0.42\textwidth]{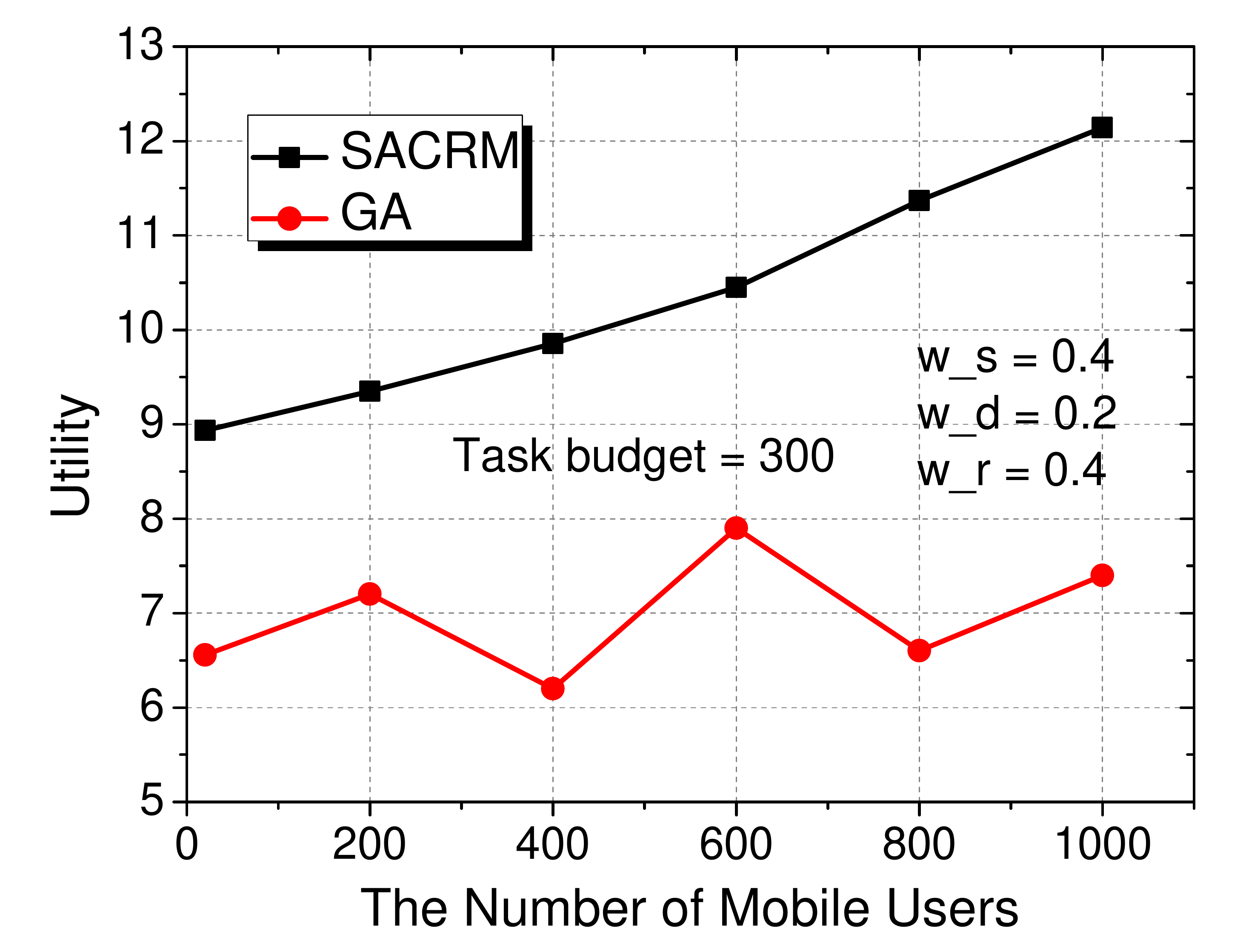}
\caption{The Scalability of SACRM.}
\label{fig.scalability}
\end{figure}

\subsection{Rewarding and Reputation Management Scheme Evaluation}
In this section, we evaluate the performance of the dynamic rewarding scheme and the reputation system in SACRM. In dynamic rewarding scheme, we set there is a participant $k$, and we have $d_t^{i} = 20$, $d_t = 40$, $w_x = 60\%$ and $b_t^{i} = 1000$. In reputation management scheme, we set the threshold of the report quality score to 0.35.

Fig.~\ref{fig.rewards_veracity} shows the allocated reward comparison under different report veracity scores. It can be seen that the allocated reward under the lower actual task delay is higher than the other one when the veracity score lies in $[0, 0.6]$. Moreover, the allocated reward increases with the increasing veracity score until it reaches the bid price. It indicates that a higher actual delay means a lower reward with a high probability under the same veracity score, and a higher veracity score brings a higher reward in most cases. Fig.~\ref{fig.rewards_delay} depicts the allocated reward comparison under different actual task delays. The change trend and the meaning of this figure are similar with these of Fig.~\ref{fig.rewards_veracity}. Note that, when the actual delay is larger than 40, the sensing report is deemed to be invalid and hence the allocated reward drops to 0. Fig.~\ref{fig.rewards_veracityanddelay} is a 3D figure that shows the reward change trend with the change of the veracity score and the actual delay of the sensing report. It indicates that the task rewards are adaptively allocated based on the report quality, which would economically stimulate the participant to improve their report quality, including the report veracity and the report delay.

\begin{figure}[!t]
\includegraphics[width=0.42\textwidth]{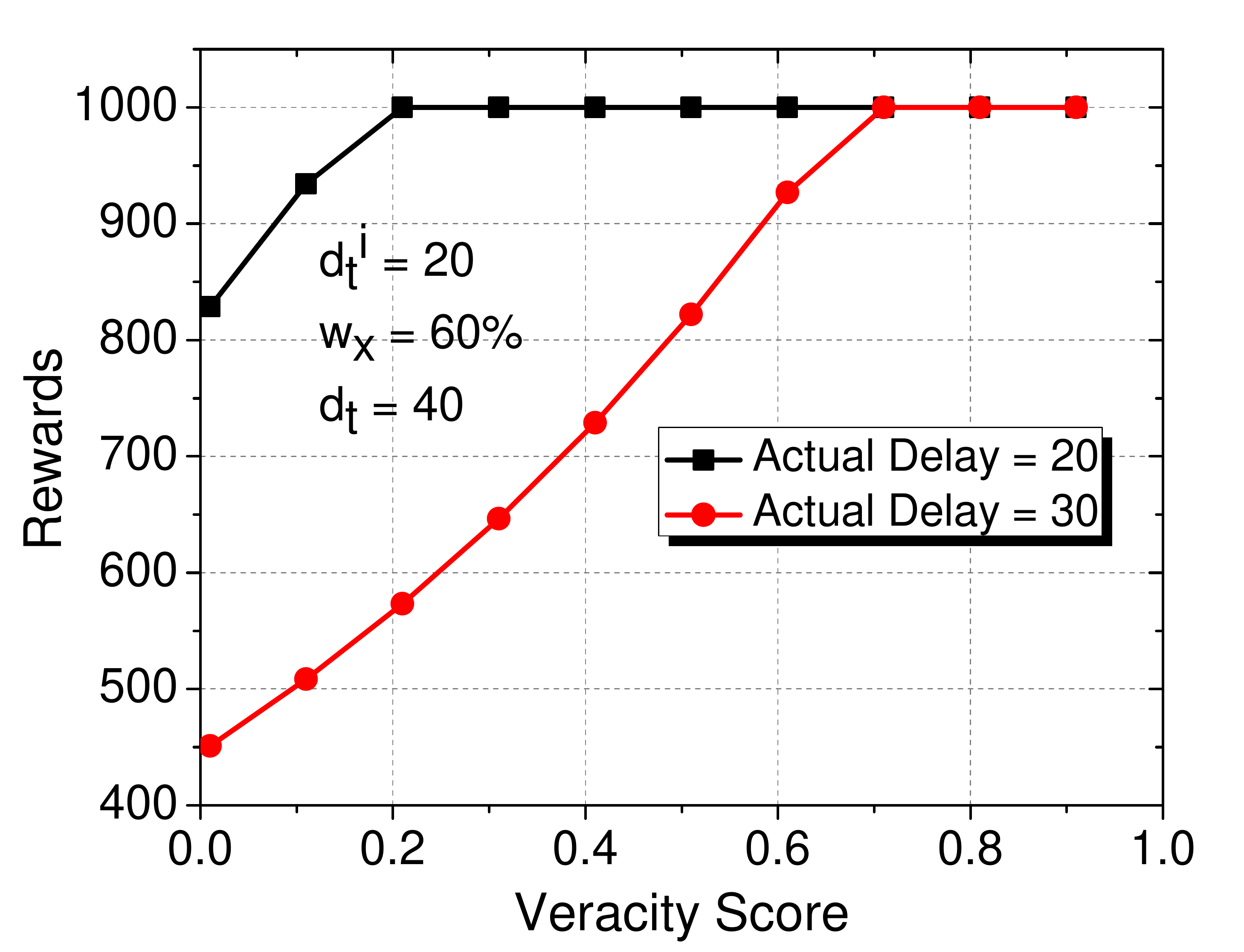}
\caption{Rewards v.s. Report Veracity.}
\label{fig.rewards_veracity}
\end{figure}
\begin{figure}[!t]
\includegraphics[width=0.42\textwidth]{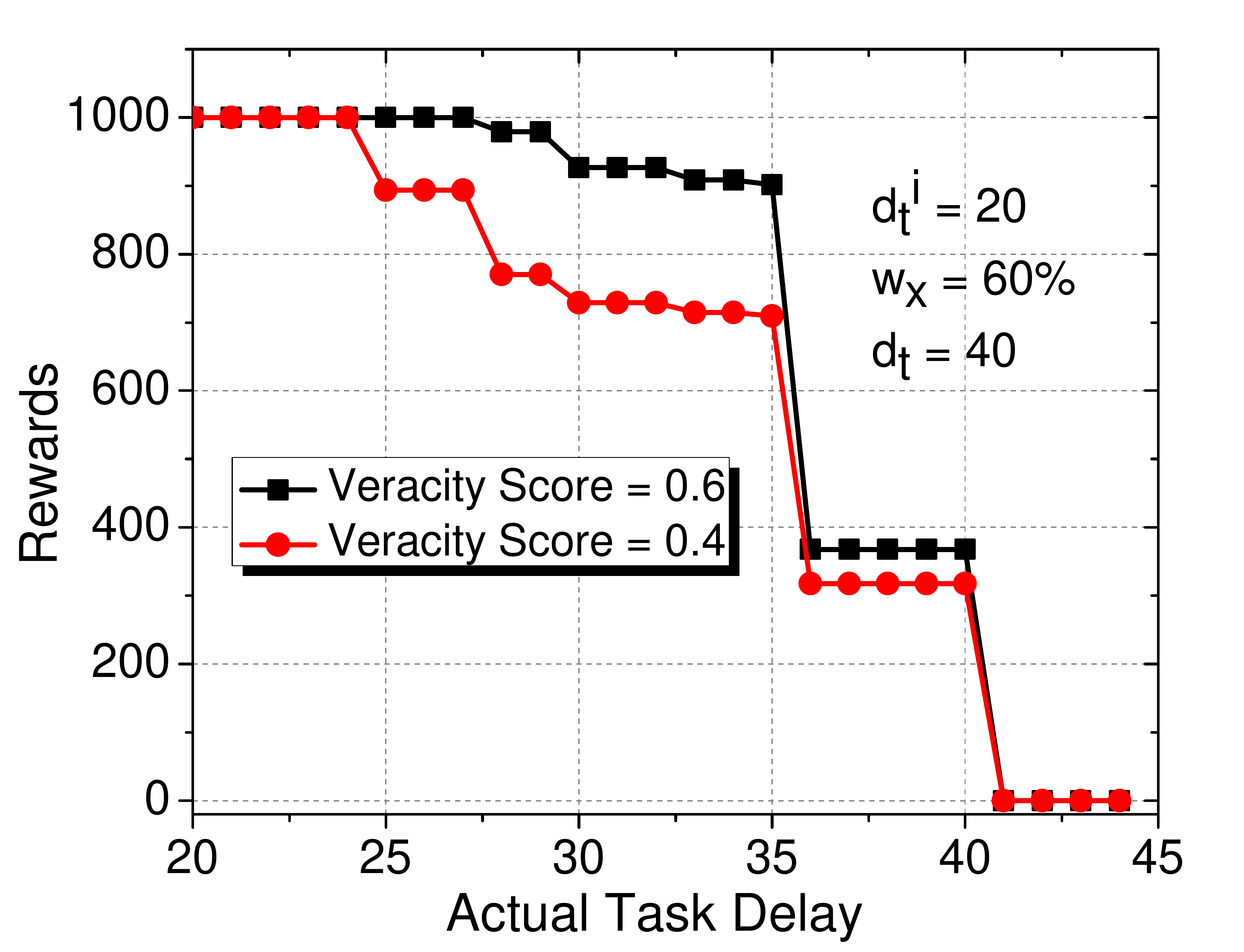}
\caption{Rewards v.s. Acutal Task Delay.}
\label{fig.rewards_delay}
\end{figure}
\begin{figure}[!t]
\includegraphics[width=0.42\textwidth]{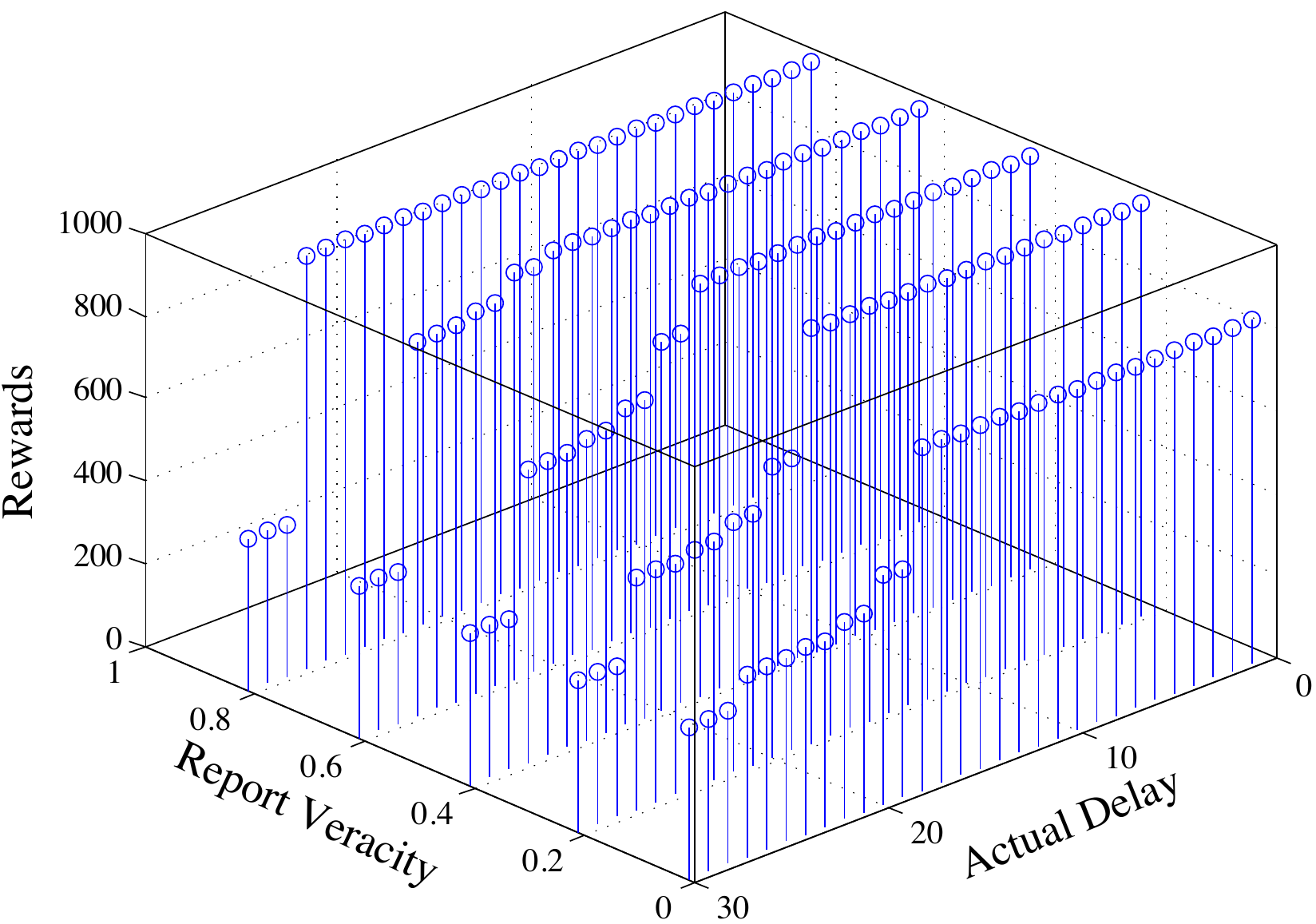}
\caption{Rewards v.s. Report Veracity and Acutal Task Delay.}
\label{fig.rewards_veracityanddelay}
\end{figure}

Fig.~\ref{fig.reputation_quality} shows the evaluated reputation value comparison under different report quality scores. It can be seen that the evaluate reputation value are $-200$ in both BPPs when the report quality score is below $0.35$. Here, BPP means the percentage of the bid price taking in the sum of all the bid prices. And the reputation value changes to be positive and increases smoothly with the increment of report quality score when the report quality score is higher than 0.35. Furthermore, at the same report quality, the participant with a lower bid price will get a higher evaluated reputation value. Fig.~\ref{fig.reputation_bidprice} depicts the evaluated reputation value comparison under different bid prices. Since the quality scores of the curves are higher than the report quality threshold, the reputation values here are both positive. However, it can be shown that when the bid price is the same, a higher report quality score obtains a relatively higher reputation value and when the quality score is the same, the reputation decreases with the increasing bid price. Fig.~\ref{fig.reputation_qualityandbidprice} is a 3D figure that shows the change trend of the reputation under the change of the participant's bid price and his report quality score. It indicates that the participants would be stimulated to increase their reputation values by improving the report quality and competitively reducing their bid prices.
\begin{figure}[!t]
\includegraphics[width=0.42\textwidth]{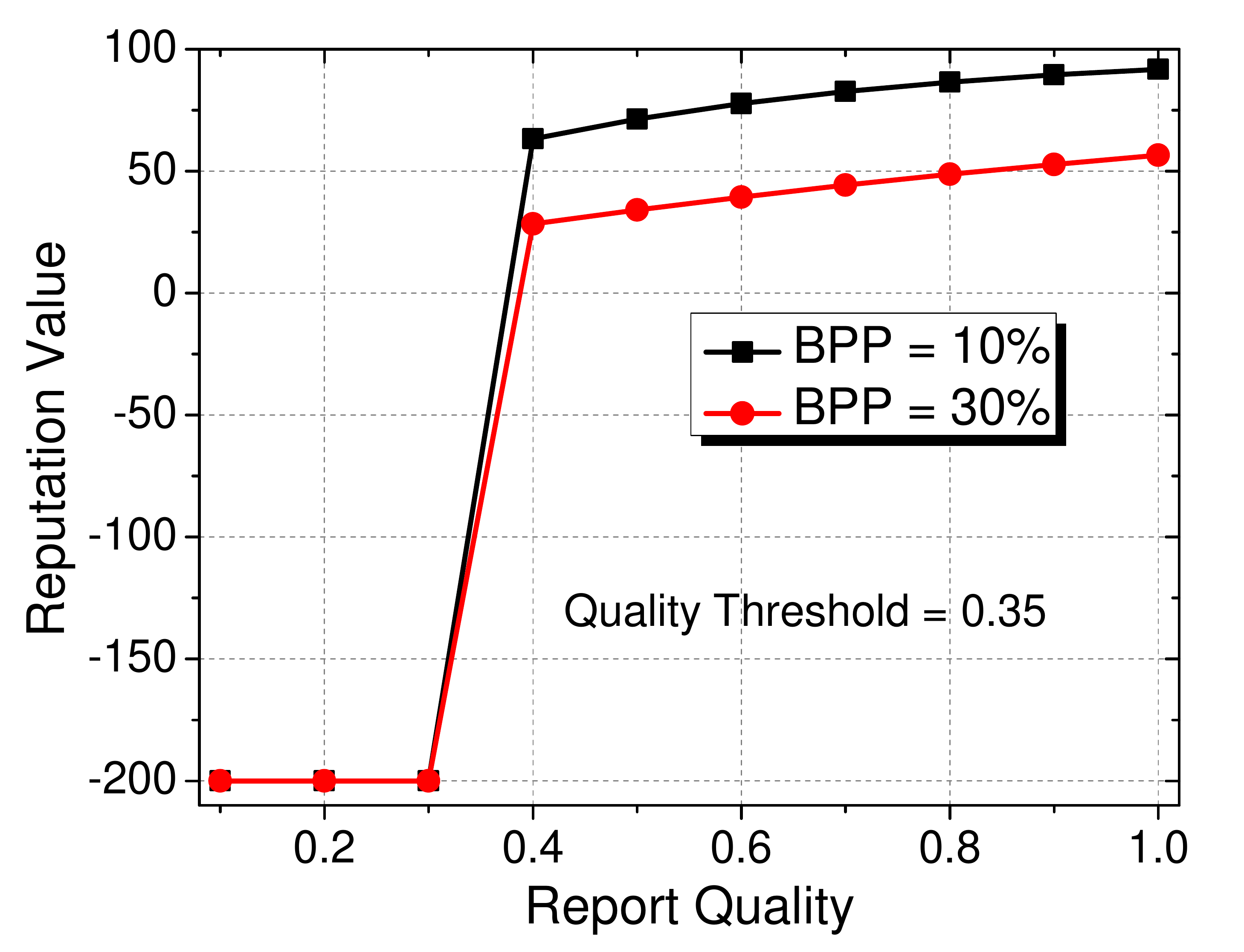}
\caption{Reputation v.s. Report Quality. (BPP means the percentage of the participant's bid price taking in the sum of all the bid prices.)}
\label{fig.reputation_quality}
\end{figure}
\begin{figure}[!t]
\includegraphics[width=0.42\textwidth]{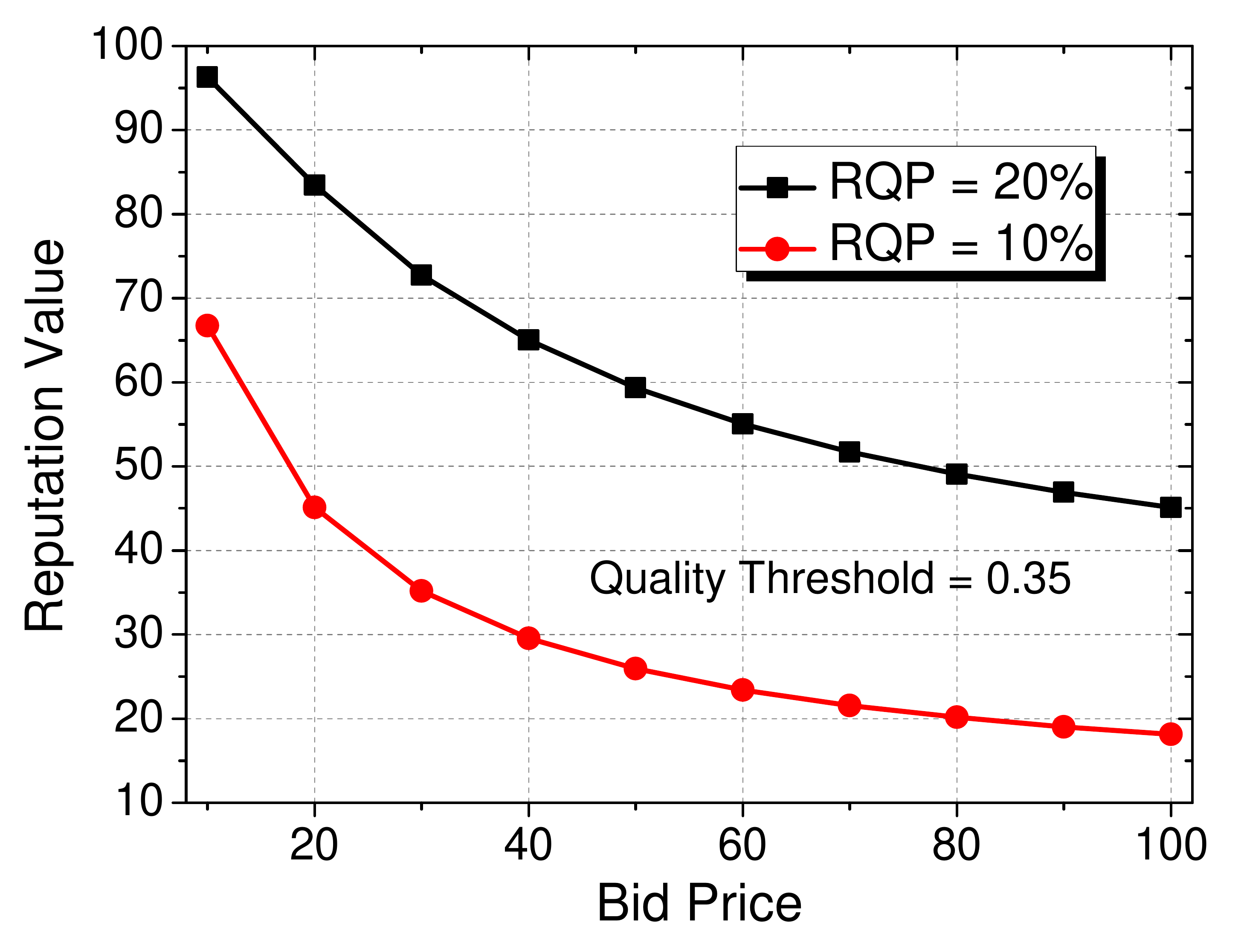}
\caption{Reputation v.s. Bidprice. (RQP means the percentage of the participant's report quality score taking in the sum of all the report quality scores.)}
\label{fig.reputation_bidprice}
\end{figure}
\begin{figure}[!t]
\includegraphics[width=0.42\textwidth]{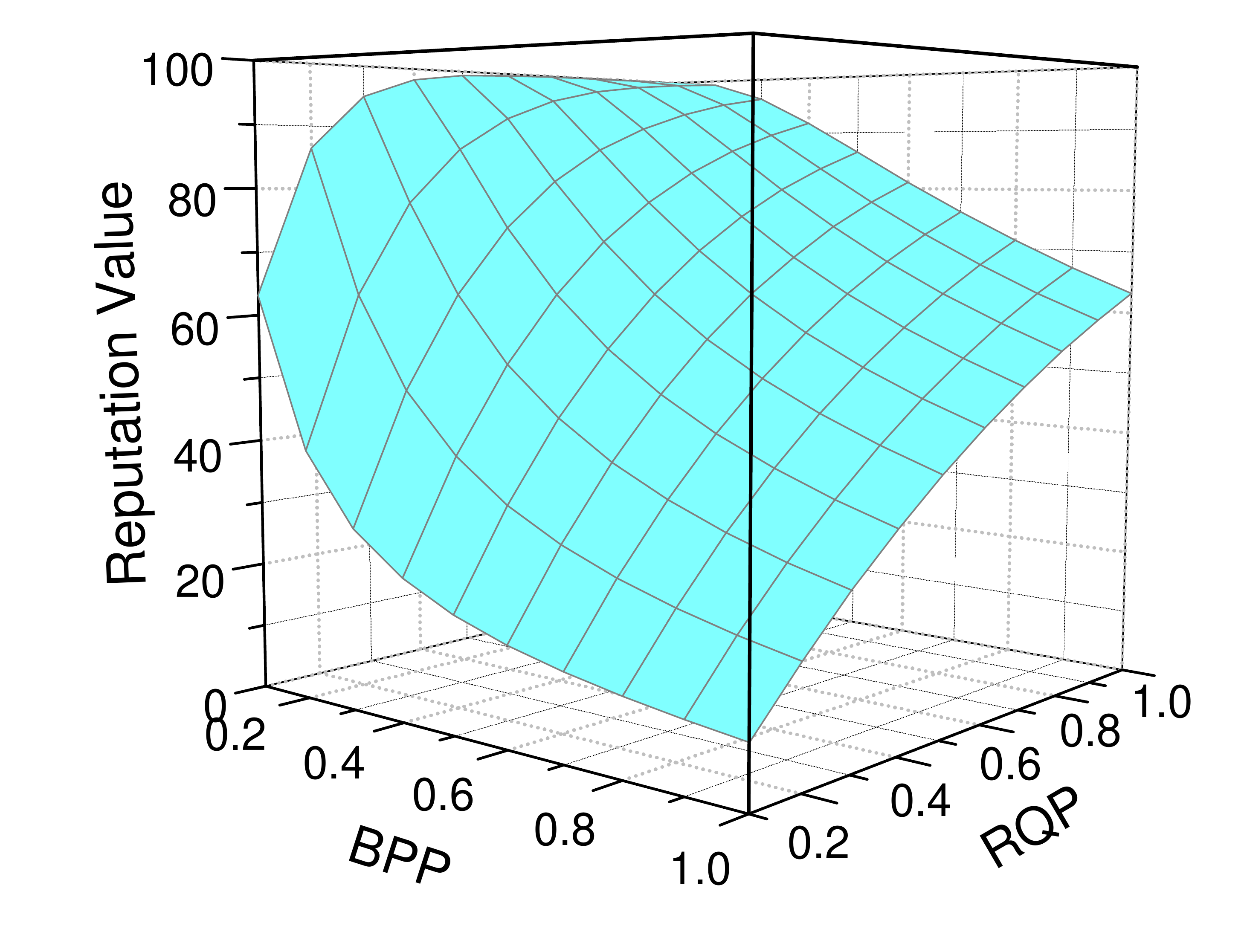}
\caption{Reputation v.s. Report Quality and Bidprice.}
\label{fig.reputation_qualityandbidprice}
\end{figure}

\section{Conclusion}
\label{sec7}
In this paper, we have proposed a \textbf{\underline{S}}ocial  \textbf{\underline{A}}ware \textbf{\underline{C}}rowdsourcing with \textbf{\underline{R}}eputation \textbf{\underline{M}}anagement (SACRM) scheme to select the well-suited participants and allocate the task rewards in mobile sensing. In participant selection, the proposed algorithms efficiently select the well-suited participants for the sensing tasks and maximize the crowdsourcing utility. Furthermore, the report quality is evaluated by the proposed report assessment scheme, and the participants are economically stimulated to improve their sensing report quality by our rewarding scheme. In addition, the proposed reputation management scheme reduces the crowdsourcing cost by introducing the cost performance ratio of the participant in reputation evaluation. Theoretical analysis and extensive simulations demonstrate that the SACRM scheme can significantly improve the crowdsourcing utility, and is effective in stimulating the participants to improve their report quality and reducing the crowdsourcing cost. In our future work, we will investigate the privacy and security issues of the SACRM scheme.

\section*{Acknowledgment}

This research work is supported by the International Science \& Technology Cooperation Program of China under Grant Number 2013DFB10070, the China Hunan Provincial Science \& Technology Program under Grant Number 2012GK4106, the Mittal Innovation Project of Central South University (No. 12MX15) and Hunan Provincial Innovation Foundation For Postgraduate, and NSERC, Canada. Ju Ren is also financially supported by the China Scholarship Council.

\bibliographystyle{elsarticle-num}
\bibliography{Reference}

\end{document}